\documentclass[smallextended]{svjour3}       
\smartqed  
\usepackage{amsmath, amssymb, mathtools, hyperref}
\usepackage[shortlabels]{enumitem}
\usepackage{graphicx}
\usepackage{caption}
\usepackage{subcaption}
\usepackage{algorithmicx}
\usepackage{algpseudocode}
\usepackage{algorithm}
\usepackage{placeins}
\usepackage{verbatim}
\usepackage{color}

\begin{document}

\title{Geometry of symmetric group-based models \thanks{Part of this work was completed, while D.Kosta was supported by a Daphne Jackson Trust Fellowship funded jointly by the Engineering and Physical Sciences Research Council and the University of Edinburgh. Part of this work was done during the Summer School on Algebra, Statistics and Combinatorics that was funded by the Aalto Science Institute research program on ``Challenges in Large Geometric Structures and Big Data''.}
}


\author{Dimitra Kosta         \and
        Kaie Kubjas 
}


\institute{D. Kosta \at
              School of Mathematics and Statistics, University of Glasgow, United Kingdom  \\
              Tel.: +44 141 330 6527\\
              \email{Dimitra.Kosta@glasgow.ac.uk}           
           \and
           K. Kubjas \at
              Department of Mathematics and Systems Analysis, Aalto University, Finland \\
              Tel.: +358 50 448 5183\\
              \email{kaie.kubjas@aalto.fi}
}

\date{Received: date / Accepted: date}

\maketitle

\begin{abstract}
Phylogenetic models have polynomial parametrization maps. For symmetric group-based models, Matsen studied the polynomial inequalities that characterize the joint probabilities in the image of these parametrizations~\cite{Matsen}. We employ this description for maximum likelihood estimation via numerical algebraic geometry. In particular, we explore an example where the maximum likelihood estimate does not exist, which would be difficult to discover without using algebraic methods. We also study the embedding problem for symmetric group-based models, i.e.~we identify which mutation matrices are matrix exponentials of rate matrices that are invariant under a group action.
\keywords{Phylogenetics \and group-based models \and maximum likelihood estimation \and numerical algebraic geometry \and algebraic statistics}
\end{abstract}

\section{Introduction}
\label{intro}
A phylogenetic tree is a rooted tree that depicts evolutionary relationships between species. A phylogenetic model is a statistical model describing the evolution of species on a phylogenetic tree. There is a discrete random variable associated with every vertex of the tree. The random variables associated with interior vertices are hidden and correspond to extinct species; the random variables associated with leaves are observed and correspond to extant species. The model parameters are the root probability and the rate or mutation matrices at the edges of the phylogenetic tree. There are different constraints on the model parameters depending on the phylogenetic model. The joint probabilities of random variables associated with leaves (leaf probabilities) are polynomials in the model parameters.

In 1987, Cavender and Felsenstein \cite{Cavender_Felsenstein} and, separately, Lake \cite{Lake}, introduced an algebraic approach to study phylogenetic models focusing on the search for phylogenetic invariants. A phylogenetic invariant of the model is a polynomial in the leaf probabilities which vanishes for every choice of model parameters. However,  phylogenetic invariants alone do not cut out all
the biologically meaningful points. One needs to include inequalities in order to obtain a complete description of the set of leaf probabilities corresponding to phylogenetic tree models.

This paper focuses on the study of group-based models which require the mutation matrices to be invariant under the action of an abelian group $G$. A symmetric group-based model assumes that the mutation matrices are symmetric. In particular, a symmetric group-based model can be a submodel of a non-symmetric group-based model with extra symmetricity conditions on mutation matrices. Phylogenetic invariants for group-based models are described in~\cite{SturmfelsSullivant}. A smaller set of phylogenetic invariants that cut out biologically meaningful points is given in~\cite{CasanellasFernandezMichalek}. Our first aim is to understand polynomial inequalities that characterize leaf probabilities for symmetric group-based models. A method for deriving the inequalities in the Fourier coordinates is given in Matsen~\cite[Proposition 3.5]{Matsen}. For the sake of completeness, we present a proof of~\cite[Proposition 3.5]{Matsen} here, and we add certain positivity constraints that appear for the Kimura 3-parameter model in~\cite{CasanellasFernandez}.  We explicitly derive the semialgebraic description of the leaf probabilities of the CFN model on the tripod tree $K_{1,3}$.

Identifying the equality and inequality characterization of the leaf probabilities is only one part of the problem. The maximum likelihood estimation aims to find parameters that maximize the likelihood of observing the data for the given phylogenetic tree and phylogenetic model. In practice, numerical methods are used to find the MLE. However, as MLE is a non-convex optimization problem, there is no guarantee of outputting the global optimum. Since phylogenetic models are not necessarily compact, the MLE might even not exist. We will use methods from computational and numerical algebraic geometry similar to~\cite{Gross_etal} to study an example for which the MLE does not exist for the CFN model on the tripod $K_{1,3}$ and a particular data vector. In this example, the global maximum is achieved when one of the model parameters goes to infinity. The nonexistence of the MLE would have been very difficult to discover without the algebraic methods that we use in this paper. One should see the example for the CFN model on the tripod $K_{1,3}$ as an illustration how to use numerical algebraic geometry for MLE in phylogenetics. It will be the subject of future work to develop a package that automatizes the computation in the phylogenetics setting, so that it can be easily used for studying further examples.

Finally, we consider the embedding problem for symmetric group-based models. Phylogenetic models are originally defined using rate matrices, e.g. matrices whose off-diagonal entries are nonnegative and the entries in each row sum to one. Mutation matrices are matrix exponentials of rate matrices. We furthermore assume that both the rate and mutation matrices satisfy the assumptions of a symmetric group-based model, i.e. they are symmetric and invariant under the action of an abelian group $G$. We call all mutation matrices that are invariant under the action of $G$ and that can be written as matrix exponentials of rate matrices that are invariant under the action of $G$, $G$-embeddable. We will characterize $G$-embeddable mutation matrices for symmetric group-based models. For the special cases of the CFN and the Kimura 3-parameter model, such characterizations are given in~\cite{Kingman,RocaFernandez}. Our characterization is for any number of states $k$ (for example $k=20$ corresponds to amino acids).

In Section~\ref{section:preliminaries}, we introduce the preliminaries of phylogenetic models and present tools from~\cite{Matsen}. Proposition~\ref{procedure:implicit_description} in Section~\ref{section:implicit_description} gives an algorithm for deriving the semialgebraic description of the leaf probabilities of a symmetric group-based model. This is mostly application of~\cite{Matsen}, however, it also considers additional positivity constraints. The main result in Section~\ref{section:embedding} is Theorem~\ref{theorem:embeddable_matrices} that gives a characterization of $G$-embeddable mutation matrices for symmetric group-based models. Finally, Algorithm~\ref{algorithm} in Section~\ref{section:MLE} outlines how to use numerical algebraic geometry to compute the MLE with probability one. This algorithm is applied on the CFN model on the tripod in Example~\ref{example:MLE}.

\section{Preliminaries}
\label{section:preliminaries}

The exposition in this section largely follows~\cite{Matsen}. A phylogenetic tree $T$ is a rooted tree with $n$ labeled leaves and it represents the genetic relationship between different
species. Its leaves correspond to current species and the internal nodes correspond to common ancestors. There is a discrete random variable $X_v$ taking $k
\in \mathbb{N}$ possible values associated to each vertex $v$ of the tree $T$. Typical values for $k$ are two, four or twenty, corresponding to a binary feature, the number of nucleotides and the number of amino acids. For example, if $k=4$, the random variable at a leaf represents the probability of observing $A,C,G$ or $T$ in the DNA of the species corresponding to the leaf.

A phylogenetic model assumes that the random variables at vertices evolve according to a Markov process, i.e. there is a transition (mutation) matrix $P^{(e)}$ associated to every edge $e$ that reflects the change in the probabilities when moving from one vertex to another. The transition matrices have the form
$$
P^{(e)}=\exp(t_e Q^{(e)}),
$$
where $\exp$ stands for matrix exponentiation, $t_e \geq 0$ represents time and $Q^{(e)}$ is a rate matrix. The non-diagonal entries of a rate matrix are non-negative and each row sums to zero. In the rest of the paper, we assume that $t_e$ is incorporated in the rate matrix $Q^{(e)}$.

Without loss of generality, we will assume that the distribution at the root is  uniform. If the distribution at the root is arbitrary, then the following procedure can be applied to reduce to the uniform case: One adds a new edge starting from the root and reroots the tree at the additional leaf. The previous root distribution is obtained by using a suitable transition matrix for the new edge. The only difference  is that the new leaf is hidden while other leaves are observed.

In this paper, we investigate group-based phylogenetic models. This means that we fix an abelian group $G$, set $k=|G|$ and assume that rate matrices are invariant under the action of $G$, i.e. there exists a vector $\psi^{(e)} \in \mathbb{R}^G$ such that $Q^{(e)}_{g,h}=\psi^{(e)}(h-g)$. We also have $P^{(e)}_{g,h}=f^{(e)}(h-g)$ for a probability vector $f^{(e)} \in \mathbb{R}^G$. The phylogenetic models we consider are symmetric, which means $Q^{(e)}_{g,h}=Q^{(e)}_{h,g}$. In the case of group-based models, this is equivalent to $\psi^{(e)}(g)=\psi^{(e)}(-g)$ and implies $f^{(e)}(g)=f^{(e)}(-g)$.

The joint probabilities $p_{i_1,\ldots,i_n}=\text{Pr}(X_1=i_1,\ldots,X_n=i_n)$ at the $n$ leaves can be written as polynomials in the root probabilities and in the entries of the mutation matrices.  Denote by ${\bf p}$ the vector of joint probabilities $p_{i_1,\ldots,i_n}$. As it is common in phylogenetic algebraic geometry, we will use the discrete Fourier transform to study the set of mutation matrices and the set of joint probabilities at the leaves for a given phylogenetic tree and a group-based model. The reason for this is that phylogenetic invariants are considerably simpler in the Fourier coordinates.

Denote by $\hat{G}$ the dual group of $G$ whose elements are the group homomorphisms from $G$ to the multiplicative group of complex numbers of magnitude one. Given a function $a:G \rightarrow \mathbb{C}$, its discrete Fourier transform is the function $\check{a}: \hat{G} \rightarrow \mathbb{C}$ defined by
$$
\check{a}(\hat{g})=\sum_{h \in G} \hat{g}(h) a(h).
$$
It is an invertible linear transformation given by the matrix $K$, where $K_{g,h}=\hat{g}(h)$. The image of the vector ${\bf p}$ of joint probabilities under the Fourier transform of $G^n$ is denoted ${\bf q}$.

The map from the entries of the rate matrices to the joint probabilities at leaves can be seen as a composition of four maps:
\begin{align} \label{four_maps}
\{\psi^{(e)}\}_{e \in E} \rightarrow \{\check{\psi}^{(e)}\}_{e \in E} \rightarrow  \{\check{f}^{(e)}\}_{e \in E} \rightarrow {\bf q} \rightarrow {\bf p}.
\end{align}

\begin{enumerate}[(1)]
\item The map from $\{\psi^{(e)}\}_{e \in E}$ to $\{\check{\psi}^{(e)}\}_{e \in E}$ is given by the discrete Fourier transform of $G$. It is an invertible linear transformation given by the matrix $K$.

\item The map from $\{\check{\psi}^{(e)}\}_{e \in E}$ to $\{\check{f}^{(e)}\}_{e \in E}$ is given by
\begin{align} \label{map2}
\check{f}^{(e)}(g)=\exp(\check{\psi}^{(e)}(g))
\end{align}
by~\cite[Lemma 2.2]{Matsen}. It is an isomorphism between $\mathbb{R}^{E \times G}$ and $\mathbb{R}_{>0}^{E \times G}$.

\item The map from $\{\check{f}^{(e)}\}_{e \in E}$ to ${\bf q}$ is given by
\begin{align}
q_{{\bf g}}=\prod_{e \in E} \check{f}^{e}(^*g_e)
\label{map3}
\end{align}
by~\cite[Theorem 3]{SzekelySteelErdos}, where $^*g_e=\sum_{i \in \Lambda(e)} g_i$ and $\Lambda(e)$ is the set of observed leaves below $e$. On the domain $\mathbb{R}_{>0}^{E \times G}$, this map is injective: \cite[Proposition 3.3 and Proposition 3.4]{Matsen} give a map from ${\bf q}$ to $\{[\check{f}^{(e)}]^2\}_{e \in E}$. Taking nonnegative square roots results in a left inverse to the map~(\ref{map3}).

\item The map from ${\bf q}$ to ${\bf p}$ is given by the inverse of the discrete Fourier transform of $G^n$. It is an invertible linear transformation given by the matrix $H^{-1}$, where $H$ is the $n$-fold Kronecker product of the matrix $K$.
\end{enumerate}

\begin{example}\label{running_example} \rm
We will consider in detail the Cavender-Farris-Neyman (CFN) model~\cite{Cavender,Farris,Neyman} on the rooted claw tree $T=K_{1,3}$. This example has been previously studied in~\cite[Example 3]{SturmfelsSullivant} and~\cite[Example 14]{Hosten_etal}. The CFN model is a group-based model with $G=\mathbb{Z}_2$ and $k=2$. Denote the root distribution by $\pi=(\pi_0,\pi_1)$ and the mutation matrices at edges $e_1,e_2,e_3$ by
\[ P^{(e_1)} = \left( \begin{array}{cc}
 \alpha^{e_1} & \beta^{e_1} \\
 \beta^{e_1} & \alpha^{e_1} \end{array} \right),
 P^{(e_2)} = \left( \begin{array}{cc}
 \alpha^{e_2} & \beta^{e_2} \\
 \beta^{e_2} & \alpha^{e_2}\end{array} \right),
 P^{(e_3)} = \left( \begin{array}{cc}
 \alpha^{e_3} & \beta^{e_3} \\
 \beta^{e_3} & \alpha^{e_3} \end{array} \right).
  \]
Since $\pi_i,\alpha^{e_i},\beta^{e_i}$ are probabilities, they are real numbers in $[0,1]$ and $\alpha^{e_i}+\beta^{e_i}=1$. Moreover, the determinant of $P^{(e_i)}$ is positive, because $P^{(e_i)}$ comes from a rate matrix $Q^{(e_i)}$. Conversely, for every $P^{(e_i)}$ satisfying these constraints, there exists a rate matrix $Q^{(e_i)}$ such that $P^{(e_i)}=\exp(t_{e_i}Q^{(e_i)})$ by~\cite[Proposition 2]{Kingman}. In Section~\ref{section:embedding}, we study constraints on mutation matrices for general symmetric group-based models.

The joint probabilities at the leaves have the parametrization
\begin{eqnarray*}
p_{000}  =  \pi_0 \alpha^{e_1} \alpha^{e_2} \alpha^{e_3} + \pi_1 \beta^{e_1} \beta^{e_2}\beta^{e_3} , & p_{001}  =  \pi_0 \alpha^{e_1} \alpha^{e_2} \beta^{e_3} + \pi_1 \beta^{e_1} \beta^{e_2} \alpha^{e_3},\\
p_{010} = \pi_0 \alpha^{e_1} \beta^{e_2} \alpha^{e_3} + \pi_1 \beta^{e_1} \alpha^{e_2} \beta^{e_3}  , & p_{011} = \pi_0 \alpha^{e_1} \beta^{e_2} \beta^{e_3} + \pi_1 \beta^{e_1} \alpha^{e_2} \alpha^{e_3},\\
p_{100} = \pi_0 \beta^{e_1} \alpha^{e_2} \alpha^{e_3} + \pi_1 \alpha^{e_1} \beta^{e_2} \beta^{e_3}  , & p_{101} = \pi_0 \beta^{e_1} \alpha^{e_2} \beta^{e_3} + \pi_1 \alpha^{e_1} \beta^{e_2} \alpha^{e_3},\\
p_{110} = \pi_0 \beta^{e_1} \beta^{e_2} \alpha^{e_3} + \pi_1 \alpha^{e_1} \alpha^{e_2} \beta^{e_3}  , & p_{111} = \pi_0 \beta^{e_1} \beta^{e_2} \beta^{e_3} + \pi_1 \alpha^{e_1} \alpha^{e_2} \alpha^{e_3}.
\end{eqnarray*}
In Section~\ref{section:implicit_description}, we characterize this model in joint probabilities $p_{ijk}$ and without parameters  $\pi_i,\alpha^{e_i},\beta^{e_i}$. This is called the implicit description of a model. It consists of polynomial equations and inequalities in $p_{ijk}$ that cut out the joint probabilities that come from a parametrization by rate matrices. In the Fourier coordinates, these equations can always be chosen to be binomials for any group-based model and tree~\cite{EvansSpeed,SzekelySteelErdos}. These binomials are characterized in~\cite[Theorem 1]{SturmfelsSullivant}. In the case of the CFN model on $K_{1,3}$, these binomials are
$$
\{ q_{001} q_{110} - q_{000} q_{111}, q_{010} q_{101} - q_{000} q_{111}, q_{100} q_{011} - q_{000} q_{111} \}.
$$
The equations defining the model in the original coordinates can be obtained by applying the Fourier transformation of $(\mathbb{Z}_2)^3$ on these binomials:
\begin{align*}
q_{000}&= p_{000}+ p_{001}+ p_{010}+ p_{011}+ p_{100}+ p_{101}+ p_{110}+ p_{111}, \\
q_{001}&= p_{000}- p_{001}+ p_{010}- p_{011}+ p_{100}- p_{101}+ p_{110}- p_{111}, \\
q_{010}&=p_{000}+ p_{001}- p_{010}- p_{011}+ p_{100}+ p_{101}- p_{110}- p_{111}, \\
q_{011}&= p_{000}- p_{001}- p_{010}+ p_{011}+ p_{100}- p_{101}- p_{110}+ p_{111}, \\
q_{100}&= p_{000}+ p_{001}+ p_{010}+ p_{011}- p_{100}- p_{101}- p_{110}- p_{111}, \\
q_{101}&= p_{000}- p_{001}+ p_{010}- p_{011}- p_{100}+ p_{101}- p_{110}+ p_{111}, \\
q_{110}&= p_{000}+ p_{001}- p_{010}- p_{011}- p_{100}- p_{101}+ p_{110}+ p_{111}, \\
q_{111}&= p_{000}- p_{001}- p_{010}+ p_{011}- p_{100}+ p_{101}+ p_{110}- p_{111}.
\end{align*}
\end{example}

Finally, we introduce basic notions from commutative algebra and algebraic geometry. A good introduction is given in~\cite{CoxLittleOShea}.  Let $R=\mathbb{R}[x_1,\ldots,x_n]$ be a polynomial ring. A subset $I \subseteq R$ is an ideal, if it is closed under addition and multiplication by scalars. The radical of an ideal $I$, denoted by $\sqrt{I}$, consists of all the polynomials $f \in R$ such that some power $f^m$ of $f$ is in $I$. Let $S$ be a set of polynomials in $R$ and let $k$ be a field. In this article, $k$ is always $\mathbb{R}$ or $\mathbb{C}$. The affine variety defined by $S$ is
$$
V(S)=\{(a_1,\ldots,a_n) \in k^n: f(a_1,\ldots,a_n)=0 \text{ for all } f \in S \}.
$$
Let $\langle f_1,\ldots,f_s \rangle$ be the ideal generated by $f_1,\ldots,f_s$, i.e.~the smallest ideal containing $f_1,\ldots,f_s$. Then
$$
V(f_1,\ldots,f_s)=V(\langle f_1,\ldots,f_s \rangle).
$$
A point of the variety $V(f_1,\ldots,f_s)$ is a smooth point, if the Jacobian of $f_1,\ldots,f_s$ has maximal possible rank. Otherwise a point of the variety is called singular. Let $T$ be a subset of $k^n$. The Zariski closure $\overline{T}$ of $T$ is the smallest affine variety containing $T$.

\section{Implicit description}\label{section:implicit_description}

Phylogenetic invariants are polynomials that vanish at joint probabilities at leaves for a given model and tree. They were introduced in~\cite{Cavender_Felsenstein,Lake} and have been characterized for group-based phylogenetic models in~\cite[Theorem 1]{SturmfelsSullivant}. A subset of them containing the biologically meaningful points  is given in~\cite{CasanellasFernandezMichalek}. This subset forms a local complete intersection and each polynomial in this subset has degree at most $|G|$. It reduces drastically the number of phylogenetic invariants used: For the Kimura 3-parameter model on a quartet tree it drops from $8002$ generators of the ideal to the $48$ polynomials described in~\cite[Example 4.9]{CasanellasFernandez}.

Besides phylogenetic invariants, polynomial inequalities are needed to give an exact characterization of joint probabilities at leaves for a given model and a tree. For the Kimura 3-parameter model, a set of inequalities is given in~\cite{CasanellasFernandez}. For general symmetric group-based models, polynomial inequalities that cut out joint probabilities at leaves are investigated in~\cite{Matsen}. This section is essentially an exposition of results in~\cite{Matsen}. There are three reasons why we present it here: For the sake of completeness, we will give a proof of~\cite[Proposition 3.5]{Matsen}, we will include the condition that Fourier coordinates need to be positive, and we will apply part of this exposition to the study of the embedding problem in Section~\ref{section:embedding}.

We recall~\cite[Propositions 3.3 and 3.4]{Matsen} that give the left inverse to the map~\ref{map3} on the domain $\mathbb{R}^{E \times G}_{>0}$.

\begin{proposition}[\cite{Matsen}, Proposition 3.3] \label{proposition:left_inverse1}
Given some leaf edge $e$, let $i$ denote the leaf vertex incident to $e$ and let $v$ be the internal vertex incident to $e$. Let $j,k$ be leaf vertices different from $i$ such that the path from $j$ to $k$ contains $v$. Let $w(g_i,g_j,g_k) \in G^n$ assign state $g_x$ to leaf $x$ for $x\in \{i,j,k\}$ and zero to all other leaf vertices. Then
$$
[\check{f}^{(e)}(h)]^2=\frac{q_{w(h,-h,0)} q_{w(-h,0,h)}}{q_{w(0,-h,h)}}.
$$
\end{proposition}

\begin{proposition}[\cite{Matsen}, Proposition 3.4] \label{proposition:left_inverse2}
Given some internal edge $e$, let the two vertices incident to $e$ be $v$ and $v'$. Let $i,j$ (respectively $i',j'$) be leaf vertices such that the path from $i$ to $j$ (respectively the path from $i'$ to $j'$) contains $v$ but not $v'$ (respectively $v'$ but not $v$). Let $z(g_i,g_j,g_{i'},g_{j'}) \in G^n$ assign state $g_x$ to leaf $x$ for $x \in \{i,j,i',j'\}$ and zero to all other leaf vertices. Then
$$
[\check{f}^{(e)}(h)]^2=\frac{q_{z(h,0,-h,0)} q_{z(0,-h,0,h)}}{q_{z(h,-h,0,0)} q_{z(0,0,-h,h)}}.
$$
\end{proposition}

\begin{proposition} \label{procedure:implicit_description}
Consider the set of $\{\psi^{(e)}\}_{e \in E}$ that satisfies $\sum_{g \in G} \psi^{(e)}(g)=0$ and $\psi^{(e)}(g) \geq 0$ for all non-zero $g \in G$. The images of this set under the maps in~(\ref{four_maps}) are:
\begin{enumerate}[(1)]
\item The constraints for $\{\check{\psi}^{(e)}\}_{e \in E}$ are obtained by substituting $\psi^{(e)}$ by $K^{-1} \check{\psi}^{(e)}$ in the constraints for $\{\psi^{(e)}\}$. In particular, this gives $\check{\psi}^{(e)}(0)=0$ and  $(K^{-1} \check{\psi}^{(e)})(g) \geq 0$ for all non-zero $g \in G$.

\item The constraints for $\{\check{f}^{(e)}\}_{e \in E}$ are  the multiplicative versions of the constraints for $\{\check{\psi}^{(e)}\}_{e \in E}$ and positivity constraints. In particular, we have $\check{f}^{(e)}(0)=1$, $(\check{f}^{(e)})^{K^{-1}_{g,:}} \geq 1$ for all non-zero $g \in G$ and $\check{f}^{(e)}(g) > 0$ for all $g \in G$. These inequalities are equivalent to $\check{f}^{(e)}(0)=1$, $(\check{f}^{(e)})^{2K^{-1}_{g,:}} \geq 1$ for all non-zero $g \in G$ and $\check{f}^{(e)}(g) > 0$ for all $g \in G$. Here we have squared the multiplicative inequalities.

\item The constraints for ${\bf q}$ are given by phylogenetic invariants, equality $q_{00\ldots 0}=1$, inequalities ${\bf q} > 0$ and inequalities that are obtained by substituting expressions for $[\check{f}^{(e)}]^2$ in Propositions~\ref{proposition:left_inverse1} and~\ref{proposition:left_inverse2} to multiplicative inequalities in the previous item.

\item The constraints for ${\bf p}$ are obtained by substituting  ${\bf q}$ by $H {\bf p}$ in the constraints for ${\bf q}$.
\end{enumerate}
\end{proposition}

Here $(\check{f}^{(e)})^{K^{-1}_{g,:}}$ denotes the Laurent monomial $\prod_{h \in G} (\check{f}^{(e)}_h)^{K^{-1}_{g,h}}$.

\begin{proof}

The constraints in items (1) and (4) are obtained, because the corresponding maps are invertible linear transformations. We will prove that the constraints in items (2) and (3) are correct.

\begin{lemma} \label{lemma:inequalities1}
The image of $\{\check{\psi}^{(e)}:\check{\psi}^{(e)}(0)=0 \text{ and } (K^{-1} \check{\psi}^{(e)})(g) \geq 0 \text{ for} \text{ all} \linebreak \text{non-zero } g \in G\}$ under the map~(\ref{map2}) is cut out by the constraints in item (2).
\end{lemma}

\begin{proof}
The positivity constraints come from the exponentiation map. Moreover, $a^T x \geq 0$ is equivalent to $e^{a^T x} \geq 1$, and since $e^{a^T x}=(e^x)^{a^T}$, it is also equivalent to $(e^x)^{a^T} \geq 1$. Hence the equalities and inequalities for $\{\check{f}^{(e)}\}_{e \in E}$ are the multiplicative versions of the equalities and inequalities for $\{\check{\psi}^{(e)}\}_{e \in E}$ together with  $\check{f}^{(e)}(g) > 0$ for all $g \in G$.
\end{proof}

\begin{lemma} \label{lemma:inequalities2}
The image of $\{\check{f}^{(e)}:\check{f}^{(e)}(0)=1,(\check{f}^{(e)})^{K^{-1}_{g,:}} \geq 1 \text{ for all non-zero } g \in G \text{ and } \check{f}^{(e)}(g) > 0 \text{ for all } g \in G\}$ under the map~(\ref{map3}) is cut out by the constraints in item (3).
\end{lemma}

Lemma~\ref{lemma:inequalities2} is very similar to~\cite[Proposition 3.5]{Matsen}, however, for the sake of completeness, we will give a proof here. We also include the positivity constraints that do not appear in~\cite[Proposition 3.5]{Matsen}. 

\begin{proof}
The inequalities ${\bf q} > 0$ are clearly valid inequalities. We will show that we do not have to additionally consider the inequalities $\check{f}^{(e)} > 0$ to construct inequalities for ${\bf q}$. Assume there is $\{\check{f}^{(e)}\}_{e \in E}$ with not all entries positive that satisfies all other inequalities in item (2) and maps to ${\bf q} > 0$. We claim that $\{|\check{f}^{(e)}|\}_{e \in E}$ also satisfies the same inequalities in item (2) and it clearly maps to the same ${\bf q}$. Indeed, since the inequalities are of the form $(\check{f}^{(e)})^{2K^{-1}_{g,:}} \geq 1$, it means that in the product $(\check{f}^{(e)})^{2K^{-1}_{g,:}}$ minus signs cancel out and hence the absolute values give the same product.

The map~(\ref{map3}) is an isomorphism between $\{\check{f}^{(e)}:\check{f}^{(e)}>0\}$ and the positive part of the Zariski closure of the image of $\{\check{f}^{(e)}:\check{f}^{(e)}>0\}$ under the map~(\ref{map3}). Indeed, let the composition of the maps in Propositions~\ref{proposition:left_inverse1} and~\ref{proposition:left_inverse2} with the map~(\ref{map3}), map ${\bf q}$ to $\{\sqrt{\frac{{\bf q}^{a_g}}{{\bf q}^{b_g}}}\}_{g \in G^n}$ for some vectors ${\bf a}_g,{\bf b}_g \in \mathbb{R}^{G^n}$. Since $q_g=\sqrt{\frac{{\bf q}^{a_g}}{{\bf q}^{b_g}}}$, or equivalently $q_g^2 {\bf q}^{b_g}={\bf q}^{a_g}$, for all ${\bf q}$ in the image, the same equation must be satisfied for all elements in the Zariski closure of the image. Moreover, $\sqrt{\frac{{\bf q}^{a_g}}{{\bf q}^{b_g}}}$ is well-defined on the positive part of the Zariski closure, hence we have the isomorphism. It follows that on the positive part of the Zariski closure we get the inequalities for ${\bf q}$ by substituting expressions for $[\check{f}^{(e)}]^2$ to multiplicative inequalities for $\check{f}^{(e)}$.
\end{proof}

This completes the proof that inequalities in items (1)-(4) are correct.
\end{proof}

\begin{example} \label{example:CFN_implicit_description} \rm
We will derive the implicit description of the CFN model on the rooted claw tree $T=K_{1,3}$. In addition to phylogenetic invariants in Example~\ref{running_example}, applying Proposition~\ref{procedure:implicit_description} gives the following inequalities in Fourier coordinates:
\begin{align}
&q_{000}=1,\nonumber\\
&{\bf q} >0,\nonumber\\
&\frac{q_{100} q_{010}}{q_{110}} \leq 1,
\frac{q_{110} q_{101}}{q_{011}} \leq 1,
\frac{q_{011} q_{010}}{q_{001}} \leq 1,
\frac{q_{001} q_{101}}{q_{100}} \leq 1. \label{ineq}
\end{align}
The inequality $\frac{q_{100} q_{010}}{q_{110}} \leq 1$ is for the hidden leaf corresponding to the root. Since $q_{000}=1$, we can multiply all the denominators by $q_{000}$ without changing the inequalities~(\ref{ineq}). Clearing denominators gives the following polynomial inequalities
\begin{align*}
&q_{000}=1,\\
&{\bf q} >0,\\
&q_{000} q_{110} - q_{100} q_{010} \geq 0,
q_{000} q_{011} - q_{110} q_{101} \geq 0,\\
&q_{000} q_{001} - q_{011} q_{010} \geq 0,
q_{000} q_{100} - q_{001} q_{101} \geq 0.
\end{align*}

By applying the discrete Fourier transformation, we get the implicit description of the CFN model on $K_{1,3}$ in the original coordinates
\begin{tiny}
\begin{align}
&p_{001}p_{010}-p_{000}p_{011}+p_{001}p_{100}-p_{000}p_{101}-p_{011}p_{110}-p_{101}p_{110}+p_{010}p_{111}+p_{100}p_{111}=0,\nonumber\\
&p_{001}p_{010}-p_{000}p_{011}+p_{010}p_{100}-p_{011}p_{101}-p_{000}p_{110}-p_{101}p_{110}+p_{001}p_{111}+p_{100}p_{111}=0,\nonumber\\
&p_{001}p_{100}+p_{010}p_{100}-p_{000}p_{101}-p_{011}p_{101}-p_{000}p_{110}-p_{011}p_{110}+p_{001}p_{111}+p_{010}p_{111}=0,\nonumber\\
&p_{000} + p_{001} + p_{010} + p_{011} + p_{100} + p_{101} + p_{110} + p_{111} = 1,\nonumber\\
& p_{000} - p_{001} + p_{010} - p_{011} + p_{100} - p_{101} + p_{110} - p_{111} > 0,\label{ineq:5}\\
& p_{000} + p_{001} - p_{010} - p_{011} + p_{100} + p_{101} - p_{110} - p_{111} > 0,\label{ineq:6}\\
& p_{000} - p_{001} - p_{010} + p_{011} + p_{100} - p_{101} - p_{110} + p_{111} > 0,\label{ineq:7}\\
& p_{000} + p_{001} + p_{010} + p_{011} - p_{100} - p_{101} - p_{110} - p_{111} > 0,\label{ineq:8}\\
& p_{000} - p_{001} + p_{010} - p_{011} - p_{100} + p_{101} - p_{110} + p_{111} > 0,\label{ineq:9}\\
& p_{000} + p_{001} - p_{010} - p_{011} - p_{100} - p_{101} + p_{110} + p_{111} > 0,\label{ineq:10}\\
& p_{000} - p_{001} - p_{010} + p_{011} - p_{100} + p_{101} + p_{110} - p_{111} > 0,\label{ineq:11}\\
&-p_{010}p_{100}-p_{011}p_{100}-p_{010}p_{101}-p_{011}p_{101}+p_{000}p_{110}+p_{001}p_{110}+p_{000}p
       _{111}+p_{001}p_{111}  \geq 0,\label{ineq:12}\\
&-p_{001}p_{010}+p_{000}p_{011}+p_{000}p_{100}-p_{001}p_{101}-p_{010}p_{110}-p_{101}p_{110}+p_{011}p
       _{111}+p_{100}p_{111}  \geq 0,\label{ineq:13}\\
&p_{000}p_{010}-p_{001}p_{011}+p_{010}p_{100}-p_{0
       11}p_{101}+p_{000}p_{110}+p_{100}p_{110}-p_{001}p
       _{111}-p_{101}p_{111}  \geq 0,\label{ineq:14}\\
&p_{000}p_{001}+p_{001}p_{010}+p_{000}p_{011}+p_{0
       10}p_{011}-p_{100}p_{101}-p_{101}p_{110}-p_{100}p
       _{111}-p_{110}p_{111}  \geq 0.\label{ineq:15}
\end{align}
\end{tiny}
\end{example}

\begin{remark} \rm
Identifiability of parameters of a phylogenetic model means that if for a fixed tree two sets of parameters map to the same joint probabilities at leaves, then these sets of parameters must be equal. Generic identifiability means that this statement is true with probability one. The identifiability of the CFN model was shown in~\cite[Theorem 1]{Hendy}, of the Kimura 3-parameter model in~\cite[Theorem 7]{SteelHendyPenny} and the generic identifiability of the general Markov model in~\cite{Chang}. The identifiability of any group-based model follows also from the proof of Proposition~\ref{procedure:implicit_description}, since each of the maps in~(\ref{four_maps}) is isomorphism in the region we are interested in.
\end{remark}

\begin{corollary}\label{corollary:boundary}
Consider a symmetric group-based model. Any ${\bf p}$ satisfying the implicit constraints of the model obtained by Proposition~\ref{procedure:implicit_description} that satisfies one of the inequalities with equality comes from a parametrization with an off-diagonal zero in the rate matrix $Q^{(e)}$ for some $e \in E$.
\end{corollary}

\begin{proof}
There are two different kind of inequalities in item (4) of Proposition~\ref{procedure:implicit_description}. The strict inequalities can never be satisfied with equality. The non-strict inequalities in each step are obtained by substituting the inverse map to the inequalities in the previous step. Hence ${\bf p}$ satisfies one of the non-strict inequalities with equality if and only if it has a preimage $\{\psi^{(e)}\}_{e \in E}$ that satisfies one of the inequalities $\psi^{(e)}(g) \geq 0$ with equality.
\end{proof}

\begin{example} \rm
We consider the CFN model. A joint probability vector ${\bf p}$ satisfying the assumptions of Corollary~\ref{corollary:boundary}, has in its parametrization the rate matrix $Q^{(e)}=\begin{pmatrix}
0 & 0\\
0 & 0
\end{pmatrix}$ for some $e \in E$. The transition matrix corresponding to the same edge is
$P^{(e)}=\begin{pmatrix}
1 & 0\\
0 & 1
\end{pmatrix}$.
\end{example}

\section{Embedding problem}\label{section:embedding}

Deciding whether a stochastic matrix $P$ is the matrix exponential of a rate matrix $Q$, is called the embedding problem. A stochastic matrix that can be represented as the matrix exponential of a rate matrix, is said to be embeddable. An overview on embeddable stochastic matrices is given in~\cite{Davies}. We call a stochastic matrix that is invariant under the action of $G$ $G$-embeddable if it can be written as a matrix exponential of a rate matrix that are invariant under the action of $G$. The aim of this section is to give an exact characterization of $G$-embeddable stochastic matrices for any symmetric group-based model.

\begin{example} \label{example:mutation_matrices_CFN} \rm
Without loss of generality, one may assume that a rate matrix in the CFN model takes the form
$$
Q^{(e)}=\begin{pmatrix}
-1 & 1\\
1 & -1
\end{pmatrix}.
$$
Applying $P^{(e)}=\text{exp}(t_e Q^{(e)})$ for $t_e \geq 0$ gives
$$
P^{(e)}=\frac{1}{2}\begin{pmatrix}
1+e^{-2t_e} & 1-e^{-2t_e}\\
1-e^{-2t_e} & 1+e^{-2t_e}
\end{pmatrix}.
$$
Studying this matrix, one obtains that $P^{(e)}$ is the matrix exponential of a rate matrix in the CFN model if and only if  $f^{(e)}(0)+f^{(e)}(1)=1$, $1 \geq  f^{(e)}(0) > \frac{1}{2}$ and $\frac{1}{2} >  f^{(e)}(1) \geq 0$. These conditions are equivalent to $P^{(e)}$ being a stochastic matrix invariant under the $\mathbb{Z}_2$-action and satisfying $\det(P^{(e)})>0$, or equivalently $\text{tr}(P^{(e)})>1$. This result is stated for $2\times 2$ stochastic matrices (not necessarily invariant under the $\mathbb{Z}_2$-action) in \cite[Proposition 2]{Kingman}.
\end{example}

For $n \times n$ stochastic matrices, a necessary condition for being embeddable is given in~\cite[Proposition 3]{Kingman}: The set set of embeddable matrices is relatively closed as a subset of the space of all stochastic $n \times n$ matrices with positive determinant. An exact characterization for stochastic matrices embeddable in the Kimura 3-parameter model (the group-based model with $G=\mathbb{Z}_2 \times \mathbb{Z}_2$), is presented in~\cite[Theorem 3.2]{RocaFernandez}. We present an exact characterization for stochastic matrices embeddable in any symmetric group-based model (for any number of states $k$). 

\begin{theorem}\label{theorem:embeddable_matrices}
Fix a group $G$ and let $K$ be the matrix of the Fourier transform of $G$. Let $P^{(e)}_{g,h}=f^{(e)}(h-g)$ for some stochastic vector $f^{(e)} \in \mathbb{R}^G$ satisfying $f^{(e)}(g)=f^{(e)}(-g)$ for every $g \in G$. Then $P^{(e)}=\exp( Q^{(e)})$ where $Q^{(e)}_{g,h}=\psi^{(e)}(h-g)$ is a rate matrix for $\psi^{(e)}$ satisfying $\psi^{(e)}(g)=\psi^{(e)}(-g)$  if and only if
\begin{enumerate}[(1)]
\item $\langle \text{Re}(K_{g,\cdot}), f^{(e)} \rangle > 0$ for every non-zero $g \in G$;

\item
\begin{small}
$$\prod_{h \in G:\text{Re}(K^{-1}_{g,h})>0} \langle \text{Re}(K_{h,\cdot}), f^{(e)} \rangle^{|\text{Re}(K^{-1}_{g,h})|} \geq \prod_{h \in G:\text{Re}(K^{-1}_{g,h})<0} \langle \text{Re}(K_{h,\cdot}), f^{(e)} \rangle^{|\text{Re}(K^{-1}_{g,h})|}$$
\end{small}
for every non-zero $g \in G$.
\end{enumerate}
Here $\text{Re}(\cdot)$ denotes the real part of a number or a vector.
\end{theorem}

\begin{proof}
The conditions on $f^{(e)}$ follow from the conditions on $\check{f}^{(e)}$ by substituting $\check{f}^{(e)}$ by $Kf^{(e)}$. By Proposition~\ref{procedure:implicit_description} item (2), the conditions for $\check{f}^{(e)}$ are $\check{f}^{(e)}(0)=1,(\check{f}^{(e)})^{K^{-1}_{g,:}} \geq 1$ for all non-zero $g \in G$ and $\check{f}^{(e)}(g) > 0$ for all $g \in G$.

The equality $\check{f}^{(e)}(0)=1$ gives $\sum_{g \in G} f^{(e)}(g)=1$. The inequalities (1) follow from $\check{f}^{(e)}(g) > 0$ for all non-zero $g \in G$. The equation $K_{g,-h}=\overline{K_{g,h}}$ always holds for a discrete Fourier transform and $f^{(e)}(h)=f^{(e)}(-h)$ holds because of symmetry. Hence
\begin{small}
\begin{align*}
K_{g,h} f^{(e)}(h) + K_{g,-h} f^{(e)}(-h)&=K_{g,h} f^{(e)}(h) + \overline{K_{g,h}} f^{(e)}(h) \\
&= \text{Re}(K_{g,h}) f^{(e)}(h) + \text{Re} (K_{g,-h}) f^{(e)}(-h)
\end{align*}
and $\langle K_{g,\cdot}, f^{(e)} \rangle=\langle \text{Re}(K_{g,\cdot}), f^{(e)} \rangle$.
\end{small}

The inequalities (2) follow from $(\check{f}^{(e)})^{K^{-1}_{g,:}} \geq 1$ for all non-zero $g \in G$. We can consider the real part for exponents, since
\begin{small}
\begin{align*}
&(\check{f}^{(e)}(h))^{K^{-1}_{g,h}} (\check{f}^{(e)}(-h))^{K^{-1}_{g,-h}} =  (\check{f}^{(e)}(h))^{K^{-1}_{g,h}} (\check{f}^{(e)}(h))^{\overline{K^{-1}_{g,h}}}\\
=&(\check{f}^{(e)}(h))^{2\text{Re}(K^{-1}_{g,h})}=(\check{f}^{(e)}(h))^{\text{Re}(K^{-1}_{g,h})} (\check{f}^{(e)}(-h))^{\text{Re}(K^{-1}_{g,-h})}.
\end{align*}
\end{small}
\end{proof}

\begin{remark} \rm
The inequalities (1) in Theorem~\ref{theorem:embeddable_matrices} imply $\det(P^{(e)})>0$. Indeed, mutation matrices for symmetric group-based models are real symmetric matrices. Their eigenvalues are $\hat{f}^{(e)}(g)=\langle \text{Re}(K_{g,\cdot}), f^{(e)} \rangle$ and the determinant is
$$
\det(P^{(e)})=\prod_{g \in G} \hat{f}^{(e)}(g)=\prod_{g \in G} \langle \text{Re}(K_{g,\cdot}), f^{(e)} \rangle.
$$
All factors in this product are positive by inequalities (1) in Theorem~\ref{theorem:embeddable_matrices} and thus $\det(P^{(e)})>0$. More precisely, the set of $G$-embeddable matrices for a symmetric group-based model is a relatively closed subset of one connected component of the complement of $\det(P^{(e)})=0$ which is given by inequalities~(1).
\end{remark}

\begin{example} \rm
The discrete Fourier transformation for the CFN model is given by the matrix
$$
K=
\begin{pmatrix}
1 & 1 \\
1 & -1
\end{pmatrix}.
$$
By Theorem~\ref{theorem:embeddable_matrices}, the conditions on mutation matrices $P^{(e)}$ in the CFN model are
$$
f^{(e)}(0)+f^{(e)}(1)=1 \text{ and } 1 \geq f^{(e)}(0)-f^{(e)}(1)>0.
$$
An easy check verifies that these conditions are equivalent to the conditions in Example~\ref{example:mutation_matrices_CFN}.
\end{example}

\begin{example} \rm
The discrete Fourier transformation for the Kimura 3-parameter model with $G=\mathbb{Z}_2 \times \mathbb{Z}_2$ is given by the matrix
$$
K=
\begin{pmatrix}
1 & 1 & 1 & 1\\
1 & -1 & 1 & -1\\
1 & 1 & -1 & -1\\
1 & -1 & -1 &1
\end{pmatrix}.
$$
Write
\begin{align*}
w=f^{(e)}(0,0)+f^{(e)}(0,1)+f^{(e)}(1,0)+f^{(e)}(1,1),\\
x=f^{(e)}(0,0)-f^{(e)}(0,1)+f^{(e)}(1,0)-f^{(e)}(1,1),\\
y=f^{(e)}(0,0)+f^{(e)}(0,1)-f^{(e)}(1,0)-f^{(e)}(1,1),\\
z=f^{(e)}(0,0)-f^{(e)}(0,1)-f^{(e)}(1,0)+f^{(e)}(1,1).\\
\end{align*}
By Theorem~\ref{theorem:embeddable_matrices}, the conditions on the mutation matrices to be $\mathbb{Z}_2 \times \mathbb{Z}_2$-embeddable are
\begin{align}\label{Z_2xZ_2_mutation_matrix_conditions}
w=1,x>0,y>0,z>0,x \geq yz, y \geq xz, z \geq yx.
\end{align}
This characterization for the Kimura 3-parameter model is first given in~\cite[Theorem 3.2]{RocaFernandez}. Moreover, it is shown in \cite{RocaFernandez} that matrices satisfying~(\ref{Z_2xZ_2_mutation_matrix_conditions}) constitute $\frac{3}{32}$ of all the stochastic matrices invariant under the $\mathbb{Z}_2 \times \mathbb{Z}_2$ action and $\frac{1}{2}$ of all such stochastic matrices with positive eigenvalues.
\end{example}

\begin{remark} \rm
By~\cite[Corollary on page 18]{Kingman}, the map from rate matrices to mutation matrices is locally homeomorphic except possibly when the rate matrix has a pair of eigenvalues differing by a non-zero multiple of $2 \pi i$. Since for symmetric group-based models rate matrices are real symmetric, then all their eigenvalues are real and hence the map from rate matrices to mutation matrices is a homeomorphism. This can be also seen by analyzing maps in Proposition~\ref{procedure:implicit_description}. Therefore the boundaries of embeddable mutation matrices of symmetric group-based models are images of the boundaries of the rate matrices. For general Markov model, the boundaries of embeddable mutation matrices are characterized in~\cite[Propositions 5 and 6]{Kingman}.
\end{remark}

\begin{corollary}
A $G$-embeddable mutation matrix lies on the boundary of the set of $G$-embeddable mutation matrices for a symmetric group-based model if and only if it satisfies at least one of the inequalities in Theorem~\ref{theorem:embeddable_matrices} with equality.
\end{corollary}

\section{Maximum likelihood estimation} \label{section:MLE}

Let ${\bf u}=(u_{i_1,\ldots,i_n})_{(i_1,\ldots,i_n) \in G^n}$ be a vector of observations at leaves. The log-likelihood function of a phylogenetic model is
$$
l_{{\bf u}}({\bf p})=\sum_{(i_1,\ldots,i_n) \in G^n} u_{i_1,\ldots,i_n} \log p_{i_1,\ldots,i_n}.
$$
Maximum likelihood estimation aims to find  a vector of joint probabilities at leaves or model parameters (if the joint probabilities are considered as polynomials in model parameters) that lie in the model and maximize the log-likelihood function for a given observation ${\bf u}$.

\begin{example} \rm
In~\cite[Example 14]{Hosten_etal}, maximum likelihood estimation on the Zariski closure of the CFN model on $K_{1,3}$ is considered. This is the model that is defined by the equations in Example~\ref{example:CFN_implicit_description}. It is shown that for  generic data, the likelihood function has $92$ complex critical points on the model. This is called the ML degree of the model. Using tools from numerical algebraic geometry as in~\cite{Gross_etal}, one can compute the $92$ critical points and among the real critical points choose the one that gives the maximal value of the log-likelihood function.

However, the MLE can lie on the boundary of a statistical model or even not exist. Neither of this can be detected by considering only the Zariski closure of the model. We will see the latter happening for the CFN model on $K_{1,3}$  in Example~\ref{example:MLE}.
\end{example}

In practice, the MLE is solved using numerical methods, however, these methods are only guaranteed to give a local maxima of the log-likelihood function and not necessarily the global maximum. Usually one runs these methods for different starting points and chooses the output that maximizes the log-likelihood function.

An alternative is to solve the Karush-Kuhn-Tucker (KKT) conditions to find all the critical points of the log-likelihood function and then choose the real solution that maximizes the log-likelihood function. Consider the optimization problem
\begin{align}
&\max F(x) \nonumber\\
&\text{subject to} \label{optimization problem}\\
&\qquad G_i(x) \geq 0 \text{ for } i=1,\ldots,m, \nonumber \\
&\qquad H_j(x)=0 \text{ for } j=1,\ldots,l. \nonumber
\end{align}
The KKT conditions are
\begin{align}
&\nabla F(x)=\sum_{i=1}^m \mu_i \nabla G_i(x)+\sum_{j=1}^l \lambda_j \nabla H_j(x),\label{equations:KKT1}\\
&G_i(x) \geq 0 \text{ for } i=1,\ldots,m,\label{inequalities:KKT1}\\
&H_j(x)=0 \text{ for } j=1,\ldots,l,\label{equations:KKT2}\\
&\mu_i \geq 0 \text{ for } i=1,\ldots,m,\label{inequalities:KKT2}\\
&\mu_i G_i(x)=0 \text{ for } i=1,\ldots,m.\label{equations:KKT3}
\end{align}
If $x^*$ is a local optimum and the optimization problem satisfies first-order constraint qualifications, then there exist $\mu_i$, where $i=1,\ldots,m$, and $\lambda_j$, where $j=1,\ldots,l$, such that $x^*$ satisfies the KKT conditions above. One first-order constraint qualification is the constant rank constraint qualification (CRCQ) defined in~\cite{Janin}. A point satisfies the CRCQ if there is a neighborhood of the point where gradients of the equality constraints and gradients of the active inequality constraints have constant rank.

In the rest of the section, we assume that conditions~(\ref{equations:KKT1})-(\ref{equations:KKT3}) are polynomial. In this case, a point satisfies the CRCQ if it is a smooth point of the variety defined by the equality and active inequality constraints. Finding all solutions of the equations~(\ref{equations:KKT1}), (\ref{equations:KKT2}) and (\ref{equations:KKT3}) in the KKT conditions is computationally heavier than using numerical methods for finding the MLE, however it is still desirable since it provides the guarantee of finding the global maximum. Symbolic solving of the equations~(\ref{equations:KKT1}), (\ref{equations:KKT2}) and (\ref{equations:KKT3})  using Gr\"oeber basis methods is possible only for small instances. An alternative is to use numerical algebraic geometry and homotopy continuation methods that find approximations of isolated complex solutions of a system of polynomial equations with probability one. This approach is taken in~\cite{Gross_etal} for optimization problems in the life sciences. Furthermore, we suggest a ``decomposition'' of the KKT conditions into parts that makes solving them easier.

Let $L$ be the ideal generated by the equations~(\ref{equations:KKT1}),~(\ref{equations:KKT2}) and~(\ref{equations:KKT3}) in the KKT conditions. For $S \subseteq [m]$, let $L_S$ be the ideal generated by the Lagrange conditions for the optimization problem
\begin{align*}
&\max F(x)\\
&\text{subject to}\\
&\qquad G_i(x)=0 \text{ for } i \in S,\\
&\qquad H_j(x)=0 \text{ for } j=1,\ldots,l.
\end{align*}
Specifically, let $L_S$ be generated by the equations
\begin{align*}
&-\nabla F(x)+\sum_{i\in S} \mu_i \nabla G_i(x)+\sum_{j=1}^l \lambda_j \nabla H_j(x),\\
&G_i(x) \text{ for } i\in S,\\
&H_j(x) \text{ for } j=1,\ldots,l.
\end{align*}
We denote by $I_S$ the ideal generated by the constraints in the above optimization problem, i.e. $I_S=\langle G_i, H_j: i \in S, j=1,\ldots,l \rangle$.

\begin{theorem} \label{theorem:KKT_decomposition}
Let $L$ and $L_S$ be as defined above. Then
$$
\sqrt{\bigcap_{S \subseteq [m]} (L_S \cap \mathbb{C}[x])}=\sqrt{L \cap \mathbb{C}[x]},
$$
where $\sqrt{\cdot}$ denotes the radical of an ideal.
\end{theorem}

\begin{proof}
We will show that $V(L \cap \mathbb{C}[x])=V(\cap_{S \subseteq [m]} (L_S \cap \mathbb{C}[x]))$. Then it will follow that
\begin{align*}
\sqrt{L \cap \mathbb{C}[x]}=I(V(L \cap \mathbb{C}[x]))=I(V(\bigcap_{S \subseteq [m]} (L_S \cap \mathbb{C}[x])))
=\sqrt{\bigcap_{S \subseteq [m]} (L_S \cap \mathbb{C}[x])}.
\end{align*}
First take an element $(\mu, \lambda, x)$ of $V(L)$. Let $S$ be such that $g_i(x)=0$ for all $i \in S$. Then $(\mu_S, \lambda, x) \in V(L_S)$, where $\mu_S$ is the projection of $\mu$ to the coordinates in $S$. Conversely, let $(\mu_S, \lambda, x) \in V(L_S)$. Let $\mu \in \mathbb{C}^m$ be such that $\mu_i=(\mu_S)_i$ for $i\in S$ and $\mu_i=0$ otherwise. Then $(\mu, \lambda, x) \in V(L)$.

We have shown that $\pi_x (V(L))=\cup \pi_x(V(L_S))$, where $\pi_x$ is the projection of $(\mu, \lambda, x)$ or $(\mu_S, \lambda, x)$ on $x$. By the Closure Theorem~\cite[Theorem 3.2.3]{CoxLittleOShea}, $V(L \cap \mathbb{C}[x])$ is the smallest algebraic variety containing $\pi_x (V(L))$ and $V(L_S \cap \mathbb{C}[x])$ is the smallest algebraic variety containing $\pi_x(V(L_S))$. The inclusion $V(L \cap \mathbb{C}[x]) \subseteq \cup V(L_S \cap \mathbb{C}[x])$ holds, because the right hand side is a variety and contains $\cup \pi_x(V(L_S))$ and hence $\pi_x (V(L))$. On the other hand, since $\pi_x(V(L_S)) \subseteq \pi_x (V(L))$ for every $S$, also $V(L_S \cap \mathbb{C}[x]) \subseteq V(L \cap \mathbb{C}[x])$ for every $S$. Hence $V(L \cap \mathbb{C}[x])=\cup V(L_S \cap \mathbb{C}[x])=V(\cap (L_S \cap \mathbb{C}[x]))$.
\end{proof}

Theorem~\ref{theorem:KKT_decomposition} suggests Algorithm~\ref{algorithm} for solving the equations in the KKT conditions.

\FloatBarrier
\begin{algorithm}[h!]
\caption{$\,\,$ Global maximum of a polynomial optimization problem}\label{algorithm}
\begin{algorithmic}
\State {\bf Input}: An optimization problem
\begin{align*}
&\max F(x)\\
&\text{subject to}\\
&\qquad G_i(x) \geq 0 \text{ for } i=1,\ldots,m,\\
&\qquad H_j(x)=0 \text{ for } j=1,\ldots,l.
\end{align*}
\State \textit{\textbf{Step 1}: Let $\mathcal{C}=\{\}$.}
\State \textit{\textbf{Step 2}: For every $S \subseteq [m]$, if $\dim(L_S)=0$ find $V(L_S)$ and add all its elements to $\mathcal{C}$.}
\State \textit{\textbf{Step 3}: Remove the elements of $\mathcal{C}$ that are not real or do not satisfy $G_i(x) \geq 0 \text{ or } \mu_i \geq 0 \text{ for } i=1,\ldots,m$.}
\State \textit{\textbf{Step 4}: Find the element $(\mu_S^*,\lambda^*,x^*)$ of $\mathcal{C}$ that maximizes $F$.}
\State  {\textbf{Output}: The element $x^*$ from Step 4.}
\end{algorithmic}
\end{algorithm}

\begin{corollary}
If $V(L)$ is finite and the global maximum of the optimization problem~(\ref{optimization problem}) satisfies CRCQ, then Algorithm~\ref{algorithm} outputs a global maximum of the optimization problem~(\ref{optimization problem}).
\end{corollary}

\begin{proof}
Theorem~\ref{theorem:KKT_decomposition} implies that $V(L \cap \mathbb{C}[x])=\cup V(L_S \cap \mathbb{C}[x])$. The variety $V(L)$ being finite implies that  $V(L \cap \mathbb{C}[x])$ and hence all $V(L_S \cap \mathbb{C}[x])$ are finite. Hence after Step 2, the list $\mathcal{C}$ contains all solutions of the equations~(\ref{equations:KKT1}),~(\ref{equations:KKT2}) and~(\ref{equations:KKT3}) in the KKT conditions. Since the global maximum satisfies the CRCQ, it must be a solution of these equations. By choosing among the real solutions that satisfy inequalities~(\ref{inequalities:KKT1}) and~(\ref{inequalities:KKT2}) in the KKT conditions the one that maximizes the value of the cost function $F$, we get the global optimum.
\end{proof}

We are interested in the optimization problem, when the cost function is the log-likelihood function $l_u$ and the constraints are polynomials that define a statistical model. Although the equations~(\ref{equations:KKT1}) are not automatically polynomial for $F=l_u$, they can be made polynomial by multiplying the equation
$$
\frac{\partial F(x)}{\partial x_k}=\sum_{i=1}^m \mu_i \frac{\partial G_i(x)}{\partial x_k} + \sum_{j=1}^l \lambda_j \frac{\partial H_j(x)}{\partial x_k}
$$
with the variable $x_k$.

One of the reasons why the variety $V(L_S)$ in Step 2 of Algorithm~\ref{algorithm} might not be finite, is that the Lagrange conditions for MLE might be satisfied by higher-dimensional components where some variable is equal to zero. For MLE, Gross and Rodriguez have defined a modification of the Lagrange conditions, known as Lagrange likelihood equations~\cite[Definition 2]{GrossRodriguez}, whose solution set does not contain solutions with some variable equal to zero, if the original data does not contain zeros~\cite[Proposition 1]{GrossRodriguez}. However, the Lagrange likelihood equations can be applied only to homogeneous prime ideals. This motivates us to study Lagrange conditions for decompositions of ideals.

\begin{lemma} \label{lemma:decomposition1}
Assume that the ideal $J=\langle G_i: i = 1,\ldots, m \rangle$ decomposes as $J=J_1 \cap J_2$, where $J_1=\langle G^{(1)}_j: j = 1,\ldots, m_1 \rangle$ and $J_2=\langle G^{(2)}_k: k = 1,\ldots, m_2 \rangle$. If $x^*$ satisfies the Lagrange conditions for the optimization problem max $F(x)$ subject to $G_i(x)=0$ for $i=1,\ldots,m$, then $x^*$ satisfies the Lagrange conditions for the optimization problem max $F(x)$ subject to $G^{(1)}_j(x)=0$ for $j=1,\ldots,m_1$ or for the optimization problem max $F(x)$ subject to $G^{(2)}_k(x)=0$ for $k=1,\ldots,m_2$.
\end{lemma}

\begin{proof}
Since $J=J_1 \cap J_2$, we have $J=\langle G^{(1)}_j G^{(2)}_k: j=1,\ldots,m_1,k=1,\ldots,m_2 \rangle$. Hence the optimization problem max $F(x)$ subject to $G_i(x)=0$ for $i=1,\ldots,m$ is equivalent to max $F(x)$ subject to $G^{(1)}_j G^{(2)}_k(x)=0$ for $j=1,\ldots,m_1,k=1,\ldots,m_2$. The Lagrange conditions for the latter optimization problem are
\begin{align*}
\frac{\partial F}{\partial x} &= \sum_{j,k} \lambda_{jk} (\frac{\partial G^{(1)}_j}{\partial x} G^{(2)}_k + \frac{\partial G^{(2)}_k}{\partial x} G^{(1)}_j)\\
&=\sum_{j} \frac{\partial G^{(1)}_j}{\partial x} (\sum_k \lambda_{jk} G^{(2)}_k)+\sum_{k} \frac{\partial G^{(2)}_k}{\partial x} (\sum_j \lambda_{jk} G^{(1)}_j),\\
&G^{(1)}_j G^{(2)}_k=0 \text{ for } j=1,\ldots,m_1,k=1,\ldots,m_2.
\end{align*}

If there exists $k$ such that $G^{(2)}_k(x^*) \neq 0$, then we must have $G^{(1)}_j(x^*)=0$ for $j=1,\ldots,m_1$. Hence $x^*$ satisfies
\begin{align*}
\frac{\partial F}{\partial x} &=\sum_{j} \frac{\partial G^{(1)}_j}{\partial x} (\sum_k \lambda_{jk} G^{(2)}_k)+\sum_{k} \frac{\partial G^{(2)}_k}{\partial x} (\sum_j \lambda_{jk} G^{(1)}_j)\\
&=\sum_{j} \frac{\partial G^{(1)}_j}{\partial x} (\sum_k \lambda_{jk} G^{(2)}_k),\\
&G^{(1)}_j=0 \text{ for } j=1,\ldots,m_1.
\end{align*}
Defining $\lambda^{(1)}_j=\sum_k \lambda_{jk} G^{(2)}_k$, we see that $x^*$ satisfies Lagrange conditions for the optimization problem max $F(x)$ subject to $G^{(1)}_j(x)=0$ for $j=1,\ldots,m_1$. Otherwise $G^{(2)}_k(x^*)=0$ for $k=1,\ldots,m_2$ and $x^*$ satisfies Lagrange conditions for the optimization problem max $F(x)$ subject to $G^{(2)}_k(x)=0$ for $k=1,\ldots,m_2$.
\end{proof}

\begin{lemma} \label{lemma:decomposition2}
Let $J=J_1 \cap J_2$ and $K=K_1 \cap K_2$. If $x^*$ satisfies the Lagrange conditions for the optimization problem max $F(x)$ subject to the generators of $J+K$, then $x^*$ satisfies the Lagrange conditions for one of the optimization problems max $F(x)$ subject to the generators of $J_j+K_k$, where $j,k \in \{1,2\}$.
\end{lemma}

\begin{proof}
Assume $J_1=\langle G^{(1)}_j: j = 1,\ldots, m_1 \rangle$, $J_2=\langle G^{(2)}_k: k = 1,\ldots, m_2 \rangle$, $K_1=\langle H^{(1)}_j: j = 1,\ldots, n_1 \rangle$ and $K_2=\langle H^{(2)}_k: k = 1,\ldots, n_2 \rangle$. Then $J=\langle G^{(1)}_j G^{(2)}_k: j=1,\ldots,m_1,k=1,\ldots,m_2 \rangle$ and $K=\langle H^{(1)}_j H^{(2)}_k: j=1,\ldots,n_1,k=1,\ldots,n_2 \rangle$. The Lagrange conditions for the generators of $J+K$ are
\begin{align*}
\frac{\partial F}{\partial x} &= \sum_{j} \frac{\partial G^{(1)}_j}{\partial x} (\sum_k \lambda_{jk} G^{(2)}_k)+\sum_{k} \frac{\partial G^{(2)}_k}{\partial x} (\sum_j \lambda_{jk} G^{(1)}_j)\\
& + \sum_{j} \frac{\partial H^{(1)}_j}{\partial x} (\sum_k \mu_{jk} H^{(2)}_k)+\sum_{k} \frac{\partial H^{(2)}_k}{\partial x} (\sum_j \mu_{jk} H^{(1)}_j),\\
&G^{(1)}_j G^{(2)}_k =0 \text{ for } j=1,\ldots,m_1,k=1,\ldots,m_2,\\
&H^{(1)}_j H^{(2)}_k =0 \text{ for } j=1,\ldots,n_1,k=1,\ldots,n_2.
\end{align*}

If there exists $k_1$ such that $G^{(2)}_{k_1}(x^*) \neq 0$ and $k_2$ such that $H^{(2)}_{k_2}(x^*) \neq 0$, then we must have $G^{(1)}_j(x^*)=0$ for $j=1,\ldots,m_1$ and $H^{(1)}_j(x^*)=0$ for $j=1,\ldots,n_1$. Hence $x^*$ satisfies
\begin{align*}
\frac{\partial F}{\partial x} &= \sum_{j} \frac{\partial G^{(1)}_j}{\partial x} (\sum_k \lambda_{jk} G^{(2)}_k)
 + \sum_{j} \frac{\partial H^{(1)}_j}{\partial x} (\sum_k \mu_{jk} H^{(2)}_k),\\
&G^{(1)}_j =0 \text{ for } j=1,\ldots,m_1,\\
&H^{(1)}_j =0 \text{ for } j=1,\ldots,n_1.
\end{align*}
Defining $\lambda^{(1)}_j=\sum_k \lambda_{jk} G^{(2)}_k$ and $\mu^{(1)}_j=\sum_k \lambda_{jk} H^{(2)}_k$, we see that $x^*$ satisfies Lagrange conditions for the optimization problem max $F(x)$ subject to the generators of $J_1+K_1$. If $G^{(2)}_k(x^*)=0$ for all $k$ and/or $H^{(2)}_k(x^*)=0$ for all $k$, then we get other combinations $J_1+K_2$, $J_2+K_1$ or $J_2+K_2$.
\end{proof}

Lemma~\ref{lemma:decomposition1} suggests that if $S$ is a singleton in Step 2 of Algorithm~\ref{algorithm}, then we can replace the ideal $L_S$ of Lagrange conditions for $I_S$ by the ideals of Lagrange conditions for minimal primes of $I_S$. If $S=\{i_1,\ldots,i_{|S|}\}$, then $I_S=I_{\{i_1\}}+\ldots+I_{\{i_{|S|}\}}$. Hence by Lemmas~\ref{lemma:decomposition1} and~\ref{lemma:decomposition2}, we can replace the ideal $L_S$ by the ideals of Lagrange conditions for the sum of minimal primes of $I_{\{i_j\}}$, where $1\leq j \leq |S|$.

\begin{remark} \label{remark:components_with_zeros} \rm
One can ignore all the components where one of the constraints is $x_k=0$ or the sum of some variables is zero. If one of the variables is zero, then the value of the log-likelihood function is $-\infty$. If the sum of some variables is zero, then all of them have to be zero, because none of them can be negative.
\end{remark}

\begin{remark} \rm
In practice, it is crucial to know the ML degree, i.e. the degree of the ideal of KKT or Lagrange conditions. If the ideal of KKT or Lagrange conditions has a finite number of solutions, then the number of solutions is equal to the ML degree. Although in theory, polynomial homotopy continuation finds all solutions of a system of polynomial equations with probability one, in practice, this can depend on the settings of the program. Without knowing the ML degree, there is no guarantee that any numerical method finds all critical points. For the CFN model on $K_{1,3}$, we experimented with \texttt{Bertini}~\cite{BHSW06}, \texttt{NumericalAlgebraicGeometry} package in \texttt{Macaulay2}~\cite{Leykin} and \texttt{PHCpack}~\cite{Verschelde}. Only \texttt{PHCpack} found all $92$ solutions with initial settings and also running times of different programs differed by several hours.
\end{remark}

\begin{example} \label{example:MLE} \rm
We aim to compute the MLE for the CFN model on $K_{1,3}$ and the data vector $(100,11,85,55, 56, 7, 75, 8)$. To do so, we relax the implicit characterization of the CFN model on $K_{1,3}$ given in Example~\ref{example:CFN_implicit_description} replacing strict inequalities by non-strict inequalities. We apply the modified version of Algorithm~\ref{algorithm} described after Lemma~\ref{lemma:decomposition2}. We automatically remove all ideals that contain a variable $p_{ijk}$, a sum of two or four such variables. The code for this example can be found at the link:
\begin{center}
\texttt{https://github.com/kaiekubjas/phylogenetics}
\end{center}

As a result we obtain $44$ ideals summarized in Table~\ref{table:boundary_components}. The first row of this table corresponds to the Zariski closure of the CFN model on $K_{1,3}$. It has degree $92$ which agrees with the ML degree 92 computed in~\cite[Example 14]{Hosten_etal}. However, to find the MLE one has to consider critical points of the likelihood function in the interior and on all the boundary components, in total $167$ of them. We compute all the $167$ complex critical points using numerical algebraic geometry software \texttt{PHCpack}. Out of the $167$ complex critical points $97$ are real and $49$ are positive. We list the seven points among them that have the highest value of the log-likelihood function in Table~\ref{table:critical_points_with_highest_loglikelihood_values}.

\begin{table}
\begin{center}
\caption[Table caption text]{Table summarizing different boundary components}
\label{table:boundary_components}
  \begin{tabular}{ | c | c | c | }
    \hline
    dim $I$ & degree $L$ & \# of ideals \\ \hline
    5 & 92 & 1 \\ \hline
    4 & 9 & 4 \\ \hline
    4 & 1 & 8 \\ \hline
    3 & 1 & 24 \\ \hline
    2 & 1 & 6 \\ \hline
    1 & 1 & 1 \\ \hline \hline
    Total & 167 & 44 \\
    \hline
  \end{tabular}
  \end{center}
  \end{table}

\begin{table}
\begin{center}
\caption[Table caption text]{Critical points with highest values of the log-likelihood function}
\label{table:critical_points_with_highest_loglikelihood_values}
\begin{tabular}{|c|c|c|}
\hline
  {\bf p} & $l_u$ & failure\\ \hline
  (0.214,0.067,0.238,0.092,0.158,0.012,0.183,0.036) & -0.0726 &  (\ref{ineq:6})\\ \hline
  (0.216,0.066,0.247,0.090,0.153,0.012,0.180,0.036) & -0.0726 &   (\ref{ineq:6}),(\ref{ineq:7}),(\ref{ineq:10}),(\ref{ineq:11}),(\ref{ineq:13}) \\ \hline
 (0.233,0.083,0.233,0.083,0.165,0.019,0.165,0.019) & -0.0729 &  (\ref{ineq:12}*) \\ \hline
 (0.231,0.084,0.231,0.084,0.166,0.019,0.166,0.019) & -0.0729 &  (\ref{ineq:12}*) \\ \hline
 (0.221,0.057,0.283,0.073,0.128,0.033,0.164,0.042) & -0.0734 &  (\ref{ineq:6}),(\ref{ineq:7}),(\ref{ineq:10}),(\ref{ineq:11}),(\ref{ineq:13}) \\ \hline
 (0.208,0.042,0.173,0.077,0.235,0.015,0.199,0.051) & -0.0737 &  (\ref{ineq:9})\\ \hline
 (0.252,0.065,0.252,0.065,0.146,0.038,0.146,0.038) & -0.0737 & MLE \\
  \hline
  \end{tabular}
  \end{center}
\end{table}

The first, second, fifth and sixth do not satisfy the relaxed inequalities and hence are not in the model. The third and fourth satisfy the inequalities of the relaxed model and hence the third one is the MLE of the model we defined at the beginning of this example. However, neither of them come from a parametrization by rate matrices. Furthermore, neither of them come from a parametrization by rate matrices where time is allowed to go to infinity. This would imply that some $\check{f}^{(e)}(g)$ and $q_{ijk}$ are zero. For this reason we have to consider a larger set of inequalities than in Proposition~\ref{proposition:left_inverse1}. Instead of picking any $j$ and $k$ such that $v$ is on the path between them, we have to consider all such $j$ and $k$. Although the third and fourth solution both satisfy in the Fourier coordinates the inequality
$$
q_{000} q_{110}-q_{100} q_{010} \geq 0,
$$
they both fail the inequality
$$
q_{000} q_{101}-q_{100} q_{001} \geq 0,
$$
which is obtained by choosing a different $k$ in Proposition~\ref{proposition:left_inverse1}. If none of $q_{ijk}$ would be zero, then if one of the two inequalities is satisfied, also the other one is satisfied. This does not have to be true anymore when some $q_{ijk}=0$.

The seventh critical point is in the image of the following parameters:
\begin{align*}
&\psi^{(e_{\text{root}})}=(0,0),\psi^{(e_1)}=(-0.665,0.665),\\
&\psi^{(e_2)}=(-\infty,\infty),\psi^{(e_3)}=(-0.262,0.262).
\end{align*}
This implies that the MLE for the CFN model on $K_{1,3}$ and the data vector $(100,11,85,55, 56, 7, 75, 8)$ does not exist -- the global maximum of the log-likelihood function is achieved when we allow one of the parameters to go to infinity. Strictly speaking this statement is true for the set of points in the model that satisfy CRCQ. We believe that for random data the global maximum will satisfy CRCQ with probability one. When we run the same optimization problem in Mathematica, then we get a solution with similar value for the log-likelihood function and all parameters besides $\psi^{(e_2)}$, which is equal to $\psi^{(e_2)}=(-8.983,8.983)$. Without having the implicit description of the CFN model on $K_{1,3}$ and using numerical algebraic geometry to study the MLE, it would be very difficult to say that the MLE does not exist.
\end{example}

\bigskip

\begin{small}

\noindent
{\bf Acknowledgments.}

We thank Elizabeth Allman, Taylor Brysiewicz, Marta Casanellas, Jes\'us Fern\'andez-S\'anchez, Serkan Hosten, Jordi Roca-Lacostena, Bernd Sturmfels, and Piotr Zwiernik for helpful discussions and comments on the first version of this paper.

\smallskip

\end{small}




\end{document}